\documentclass[aps,pra,twocolumn,10pt,superscriptaddress,longbibliography, floatfix]{revtex4-1}
\usepackage[utf8]{inputenc}  
\usepackage[T1]{fontenc}     
\usepackage[british]{babel}  
\usepackage[scaled=0.86]{berasans}  
\usepackage[colorlinks=true,allcolors=blue]{hyperref}  
\usepackage{graphicx} 
\usepackage[babel]{microtype}  
\usepackage{amsmath,amssymb,amsthm,bm,amsfonts,mathrsfs,bbm} 
\usepackage{setspace}
\usepackage{braket}
\usepackage{blochsphere}
\usepackage{pgfplots}
\usepackage{float}
\usetikzlibrary{calc,3d,shapes, pgfplots.external, intersections}
\usepackage[bitstream-charter]{mathdesign}
\usepackage{tikz}
\usepackage{mathtools}
\usetikzlibrary{arrows.meta}
\usepackage{csquotes}
\usepackage{multirow}
\usepackage{adjustbox}

\usepackage{xspace}  
\usepackage{pgfplots}
\usepackage{verbatim}
\usepackage[bitstream-charter]{mathdesign}
\usepackage{todonotes}
\usepackage{xfrac}
\usepackage[nice]{nicefrac}
\usepackage{caption}
\usepackage{subcaption}
\usepackage[normalem]{ulem}
\usepackage{array, tabularx, boldline}
\usepackage{cellspace}
\usepackage{makecell}
\usepackage[normalem]{ulem}
\usepackage{color}
\usepackage[mathscr]{euscript}

\newtheorem{theorem}{Theorem}

\newtheorem{lemma}[theorem]{Lemma}

\makeatletter
    \renewcommand\@make@capt@title[2]{%
     \@ifx@empty\float@link{\@firstofone}{\expandafter\href\expandafter{\float@link}}%
      {\textbf{#1}}\@caption@fignum@sep#2\quad}%
    \makeatother
   
\makeatletter 
\renewcommand{\fnum@figure}{\textbf{Figure~\thefigure}}
\makeatother

\usepackage{accents}
\usepackage{gensymb} 
\usepackage{dcolumn}
\usepackage{bm}
\usepackage{yfonts}
\usepackage{mathtools}
\usepackage{comment}

\usepackage{notes2bib}
\bibnotesetup{
note-name = ,
use-sort-key = false
}

\newcommand\id{\leavevmode\hbox{\small1\kern-3.3pt\normalsize1}}
\newcommand{\ketbra}[2]{\left|#1\middle\rangle\middle\langle#2\right|}

\begin{document}


\title{Polarization correction towards satellite-based QKD without an active feedback}

\author{Sourav Chatterjee}
\affiliation{Raman Research Institute, C. V. Raman Avenue, Sadashivanagar, Bengaluru, Karnataka 560080, India}
\author{Kaumudibikash Goswami}
\affiliation{Raman Research Institute, C. V. Raman Avenue, Sadashivanagar, Bengaluru, Karnataka 560080, India}
\author{Rishab Chatterjee}
\affiliation{Raman Research Institute, C. V. Raman Avenue, Sadashivanagar, Bengaluru, Karnataka 560080, India}
\author{Urbasi Sinha}
\email[]{usinha@rri.res.in}
\affiliation{Raman Research Institute, C. V. Raman Avenue, Sadashivanagar, Bengaluru, Karnataka 560080, India}

\begin{abstract}
    Quantum key distribution (QKD) is a cryptographic protocol to enable two parties to share a secure key string, which can be used in one-time pad cryptosystem. There has been an ongoing surge of interest in implementing long-haul photonic-implementation of QKD protocols. However, the endeavour is challenging in many aspects. In particular, one of the major challenges is the polarization degree of freedom of single-photons getting affected while transmission through optical fibres, or atmospheric turbulence. Conventionally, an active feedback-based mechanism is employed to achieve real-time polarization tracking. In this work, we propose an alternative approach where we first perform a state tomography to reconstruct the output density matrix. We then evaluate the optimal measurement bases at Bob's end that leads to the maximum (anti-)correlation in the measurement outcomes of both parties. As a proof-of-principle demonstration, we implement an in-lab BBM92 protocol --- a particular variant of a QKD protocol using quantum entanglement as a resource --- to exemplify the performance of our technique. We experimentally generate polarization-entangled photon pairs having $94\%$ fidelity with $\ket{\psi}_1 = 1/\sqrt{2}\,(\ket{HV}+\ket{VH})$ state and a concurrence of $0.92$. By considering a representative 1 ns coincidence window span in our implementation involving a novel alternative to an active feedback-based mechanism, we are able to achieve a quantum-bit-error-rate (QBER) of $\approx 5\%$, and a key rate of $\approx 35$ Kbps. The performance of our implemented protocol is independent of the local polarization rotations through optical fibres. We have also developed an algorithmic approach to optimize the trade-off between the key rate and QBER. Our approach obviates the need for active polarization tracking. Our method is also applicable to entanglement-based QKD demonstrations using partially mixed as well as non-maximally entangled states, and extends to single-photon implementations over fibre channels.
\end{abstract}

\maketitle

\section{Introduction}

State-of-the-art classical (public key) cryptosystems, based upon Rivest-Shamir-Adleman (RSA) algorithm\,\cite{rivest1983cryptographic}, offers security dependent upon computational assumptions, which can be easily broken once large scale quantum computers become available\,\cite{xu2020secure}. The solution to this threat is offered by a relatively new cryptographic primitive: quantum key distribution (QKD). It offers information-theoretically secure communication, i.e., free from any algorithmic or computational advancements\,\cite{chatterjee2020qkd}. The first QKD protocol was experimentally demonstrated over a 30 cm long free-space optical channel\,\cite{bennett1984ieee, bennett1992experimental}. Over the years, several sophisticated methods for performing QKD have been proposed\,\cite{ekert91, B92, bbm92, PhysRevLett.113.140501} and successfully implemented within laboratory environment\,\cite{jennewein2000quantum, ling2008e91, erven2009entanglement, diqkd_exp, chatterjee2020qkd}. Beyond the shielded lab atmosphere, on one hand several experiments have been performed to test the practical limits of wide-scale deployment of QKD between the two communicating parties, commonly known as Alice (sender) and Bob (receiver), using optical fibres\,\cite{scarani2009security, xu2020secure}. However, it has been reported that the attenuation loss and background noise suffered in fibre-based QKD transmissions prohibit achieving sufficiently large key rates beyond metropolitan-scale networks\,\cite{bedington2017progress, toyoshima_polarization-basis_2011}. On the other hand, satellite-based QKD serves as a promising technique in overcoming this transmission distance scaling issue. Hence, over the last decade, many free-space experiments have been performed to test QKD implementations with a moving platform including hot-air balloon\,\cite{wang2013direct}, truck\,\cite{bourgoin2015free}, aircraft\,\cite{pugh2017airborne, nauerth2013air}, and drone\,\cite{liu2020drone}. Furthermore, the progress in China's Quantum Experiments at Space Scale project has enabled world-wide efforts towards realizing full QKD demonstrations in free-space using orbiting satellites\,\cite{bedington2017progress, joshi2018space}.

The polarization of light is a commonly used degree of freedom to achieve the above-mentioned practical implementations of the QKD protocols. However, maintaining the polarization of light over long distance QKD protocols has practical challenges. For optical fibre based QKD protocols, the polarization state is affected due to randomly varying birefringence of the optical fibre\cite{vanwiggeren_transmission_1999, gordon_pmd_2000}. In case of free-space, although the polarization is comparatively robust against atmospheric turbulence~\cite{zhang_polarization_2017,korotkova_changes_2005,zhu_compensation-free_2021,yang_influence_2018}, the reference frame of the satellite plays a detrimental role---polarization changes according to the movement of the satellite. Hence, it is important to circumvent such polarization changes in both free-space and fibre-based QKD. Conventional mitigation techniques involve active polarization tracking devices. For instance, in Ref.~\cite{lee_robotized_2022} the authors used a robotized polarization correction based on an active control system. In Ref.~\cite{toyoshima_polarization-basis_2011}, Toyoshima et al. established a 1 km free space QKD link with an active control system based polarization tracking jitter error of 0.092\degree. In the above protocols, the authors calibrated the polarization change during the QKD session. In an alternative approach, in Ref.~\cite{ding_polarization-basis_2017}, the authors performed the polarization basis tracking using the sifted keys revealed during the QKD error-correction procedure. In case of fibre-based protocols, the polarization compensation was done in Ref.~\cite{Xavier_2009}, where the single photons were wavelength-multiplexed with two classical beams. The classical beams reveal the information regarding polarization fluctuation. Based on the polarization fluctuation, the authors used an active polarization control system to compensate for the polarization change. Fast feedback-based polarization controlling over an aerial fibre has been demonstrated in Ref.~\cite{li_field_2018}. These conventional approaches to mitigate the polarization fluctuation require active control systems. In this work, we propose an alternative solution where instead of correcting polarization fluctuation of the encoded state, we optimize the measurement bases at the receiver-end. 

Our method overcomes a number of challenges associated with active feedback systems. Firstly, while the aim of active feedback systems is to stabilize an unstable input state, often such control systems have more elements than the raw system, which incurs into additional costs. Furthermore, active control systems having more parts than the raw system, they are more prone to faults, which can lead to instability of the closed loop. In addition, such control systems often employ trial and error methods to nullify the output deviations. This approach often leads to the oscillatory response of the closed loop. We overcome the above challenges by the following procedure. We first perform a quantum state tomography at the output. Then, from the tomographically reconstructed density matrix, we evaluate the receiver's optimal choice of measurement bases such that the measurement outcomes lead to high (anti-)correlation that is required for successful key generation. Note that this approach obviates the need for maintaining the polarization state of the photons using resource-intensive control systems. The time required to evaluate the optimal measurement bases is solely dependent on the tomography of the output state, which could be significantly improved with efficient tomography techniques that offer faster convergence rates based on Bayesian learning ~\cite{Granade_2016,Evans_2022}, or machine learning-based approaches~\cite{SGT_theory,Rambach_2021,palmieri_experimental_2020}. Using an in-lab single-photon based BBM92 protocol implementation~\cite{bbm92, erven2009entanglement}, we demonstrate that our approach mitigates any performance limitations of the protocol otherwise posed upon by polarization fluctuations of the entangled photons. It is important to note, however, that our method can be generalized towards any QKD protocols. 

The other significant aspect of our work involves novel optimization techniques for single-photon based BBM92 protocol implementation \bibnotemark[optimization]. In case of the BBM92 protocol, once the communicating parties generate the time-stamps after measuring the entangled photons in the desired and undesired bases, our optimization techniques find the optimal coincidence windows required for determining the maximal (desired) and minimal noise (undesired) coincidences. To assess the performance of the BBM92 protocol, we use standard measures: quantum-bit-error-rate (QBER), and key rate. Key rate is the number of key-bits generated per second. QBER is the ratio of error rate to key rate. Using our optimization technique, we are able to achieve higher key rate while maintaining information-theoretically secure QBER ($<11\%$)\,\cite{renner2005, Preskil-11-percent}.

To demonstrate our novel polarization mitigation and optimization technique, we implement an in-lab BBM92 protocol. We first produce polarization entangled single-photon pairs using a Sagnac interferometer based type-II Spontaneous Parametric Downconversion (SPDC) source \cite{fedrizzi_wavelength-tunable_2007}. Through the optical fibres, we transmit the generated polarization-entangled single-photon pairs to two modules. The operations performed at the modules represent the operations by the communicating parties, henceforth we will refer to these modules as Alice and Bob. In a conventional BBM92 protocol, Alice and Bob agree on two out of three mutually unbiased measurement bases, $\sigma_1:\{\ket{D},\ket{A}\}$, $\sigma_2:\{\ket{R},\ket{L}\}$, and $\sigma_3:\{\ket{H},\ket{V}\}$. Here $\sigma_i$ is the Pauli operator, and the corresponding measurement bases are the eigenstates of the Pauli operator. On each single photon, the parties randomly measure the polarization in $\sigma_i$ or $\sigma_j$ bases to generate the time-stamps on which the cross-correlation is performed to detect the (anti-)correlations. However, the fibre birefringence affects the polarization states of both Alice and Bob. This jeopardizes the expected coincidence counts after the measurements. To mitigate this effect, we evaluate the optimal measurement bases at Bob's side to achieve higher coincidences. Even without optimization, irrespective of the output polarization state, our choice of optimal measurement bases allows us to achieve around $5\%$ QBER and $35$ Kbps key rate for $1$ ns coincidence window, and around $10\%$ QBER and $50$ Kbps key rate for $4$ ns coincidence window. Upon further optimization, we are also able to restrict the individual QBERs for all four measurement pairs below $11\%$ bound, while restricting the overall QBER to around $8.5\%$.
 
\section{Results}
\subsection{Theory for optimized measurement bases}\label{subsec:theory}
In this section, we describe our method to construct the optimal measurement bases to mitigate the polarization fluctuation during transmission of single photons over long distance. In practice, the polarization state of both the photons would be affected. However, the polarization fluctuation of two subsystems could be mitigated by addressing only one of the subsystems. We convey this in the following lemma where we show that two local unitary operations on each subsystem of an entangled state is equivalent to a single unitary operation in one of the subsystems.  
\begin{lemma}
The action of local unitary operations $U$ and $V$ on each subsystem of a Bell state $\ket{\psi}_i^{AB}$ is equivalent to a single local unitary operation $W{=}V\sigma_iU^T\sigma_i$ on the subsystem $B$, i.e. $(U^{A}{\otimes}V^{B})\ket{\psi}_i^{AB}{=}(\mathbb{1}^{A}{\otimes}W^{B})\ket{\psi}_i^{AB}$.  
\end{lemma}
\begin{proof}

Firstly, let us consider the Bell-state $\ket{\psi}_0{=}1/\sqrt{2}(\ket{00}{+}\ket{11})$.
It is well-known~\cite{Coecke_2010} that any unitary operation $U$ acting on one sub-system of $\ket{\psi}_0$ is equivalent to the transpose of the same unitary $U^T$ acting on the other sub-system:

\begin{align}
    (U{\otimes}\mathbb{1})\ket{\psi}_0{=}(\mathbb{1}{\otimes}U^T)\ket{\psi}_0. 
\end{align}
Hence, for two local unitary operations $U$ and $V$ on each subsystem of $\ket{\psi}_0$:
\begin{align}
    (U{\otimes}V)\ket{\psi}_0{=}(\mathbb{1}{\otimes}V).(U{\otimes}\mathbb{1})\ket{\psi}_0{=}(\mathbb{1}{\otimes}VU^T)\ket{\psi}_0. \label{Eq:phi_plus}
\end{align}
Now, let us consider other Bell-states $\{\ket{\psi}_i\}$ which are related to $\ket{\psi}_0$ by local Pauli operations $\{\sigma_i\}$:

\begin{align}
    \ket{\psi}_i{=}(\mathbb{1}{\otimes}\sigma_i)\ket{\psi}_0.\label{Eq:relation1}
\end{align}
As Pauli matrices are self-inverse, we also have

\begin{align}
    \ket{\psi}_0{=}(\mathbb{1}{\otimes}\sigma_i)\ket{\psi}_i.\label{Eq:relation2}
\end{align} 
For the sake of consistency, we assume $\sigma_0$ to be the identity operation. Now for two local unitary operations $U$ and $V$ acting on $\ket{\psi}_i$, we can write:
\begin{align}
    (U{\otimes}V)\ket{\psi}_i&{=}(U{\otimes}V\sigma_i)\ket{\psi}_0{=}(\mathbb{1}{\otimes}V\sigma_iU^T)\ket{\psi}_0 \nonumber \\
    &{=}(\mathbb{1}{\otimes}V\sigma_iU^T\sigma_i)\ket{\psi}_i{=}(\mathbb{1}{\otimes}W)\ket{\psi}_i.
\end{align}
Here, $W{=}V\sigma_iU^T\sigma_i$. The first, second and the third equalities are due to Eqs.~\eqref{Eq:relation1}, ~\eqref{Eq:phi_plus}, and ~\eqref{Eq:relation2}, respectively. This concludes our Lemma.
\end{proof}
The above lemma suggests that if we could infer the relevant unitary operation $W$, we could absorb $W$ in our measurement, i.e., if the desired measurement basis were $\{\ket{\alpha},\ket{\alpha^{\perp}}\}$, to mitigate the effect of $W$, we would have to measure in  $\{W^{\dagger}\ket{\alpha},W^{\dagger}\ket{\alpha^{\perp}}\}$ basis. To estimate the unitary, or equivalently the pure state  after the action of the unitary $W$, we first perform quantum state tomography at the output, let us assume the tomographically reconstructed density matrix is $\rho$. Next, we evaluate the nearest pure state of $\rho$, this nearest pure state will basically be the bell state affected by the unitary $W$. Next, based on the nearest pure state, we find the optimal measurement basis at Bob's end: Bob's $\{\ket{\phi_H},\ket{\phi_H^\perp}\}$ measurement basis showing the maximum (anti-)correlation with Alice's $\{\ket{H},\ket{V}\}$ measurement basis, and  Bob's $\{\ket{\phi_D},\ket{\phi_D^\perp}\}$ basis showing the maximum correlation with Alice's $\{\ket{D},\ket{A}\}$ measurement. In the next subsection, we show the method to find the nearest pure state.
\subsubsection{Nearest pure state from the eigendecomposition} \label{sec:nearest_pure}
The nearest pure state of a density matrix is, as the name suggests, the state that has the maximum overlap to the said density matrix. In our case, we are using fidelity to define the overlap. Formally, the pure state $\ket{\psi_\rho}$ is the `nearest' to the density matrix $\rho$ when their fidelity $F$ satisfies 

\begin{align}
    F =  \bra{\psi_\rho} \rho \ket{\psi_\rho} \geq \bra{\alpha}{\rho}\ket{\alpha}, \ \ \text{ for any pure state $\ket{\alpha}$}.
\end{align}

To find the nearest pure state, we perform eigendecomposition of the density matrix $\rho$. The nearest pure state would be the eigenvector corresponding to the maximum eigenvalue. Formally, if the eigendecomposition of $\rho{=}\sum_i \lambda_i\ketbra{\lambda_i}{\lambda_i}$, with $\{\lambda_i\}$ being the eigenvalues and $\{\ket{\lambda_i}\}$ being the corresponding eigenvector. We also assume that $\{\lambda _i\}$ are arranged in descending order: $\lambda_i{\geq}\lambda_j$ for $i>j$. Hence, according to our notation, $\lambda_1$ is the maximum eigenvalue and $\ket{\lambda_1}$ is the nearest pure state. In the following lemma, we will prove $\ket{\lambda_1}$ indeed has the maximum overlap with $\rho$.

\begin{lemma}
For a density matrix $\rho$ having eigendecomposition $\rho{=}\sum_i \lambda_i\ketbra{\lambda_i}{\lambda_i}$, with $\lambda_i{\geq}\lambda_j$ for $i>j$, the nearest pure state of the density matrix would be $\ket{\lambda_1}$.
\end{lemma}\label{Lemma:nearest_pure}

\begin{proof}
Let us consider an arbitrary pure state $\ket{\alpha}$ with the spectral decomposition $\ket{\alpha}{=}\sum_i a_i \ket{\lambda _i}$, with $\sum_i|a_i|^2{=}1$. The fidelity between the state $\ket{\alpha}$ and $\rho$ is:
\begin{align}
    \bra{\alpha}{\rho}\ket{\alpha}&{=}\sum_{i,j}a^*_ia_j\bra{\lambda_i}{\rho}\ket{\lambda_j}\nonumber \\
    &{=}\sum_{i,j,k}a^*_ia_j\lambda_k\langle{\lambda_i}|{\lambda_k}\rangle \langle{\lambda_k}|{\lambda_j}\rangle {=}\sum_i \lambda_i|a_i|^2.
\end{align}
As the set $\{\lambda_i\}$ are in descending order, and the set $\{|a_i|^2\}$ form a probability distribution:$\sum_i |a_i|^2{=}1$ and $0{\leq}|a_i|^2{\leq}1$, the quantity $\sum_i \lambda_i|a_i|^2$ is maximum if and only if $a_1{=}1$ and $a_{i\neq 1}=0$. In that case, $\ket{\alpha}{=}\ket{\lambda_1}$: the eigenvector corresponding to the maximum eigenvalue.
\end{proof}

\subsubsection{Optimal measurement bases for BBM92 protocol}

From our tomographically obtained density matrix $\rho^{AB}$, we find the nearest pure state $\ket{\psi}_\rho^{AB}$. We can express the nearest pure state in the form

\begin{align}
    \ket{\psi}_\rho^{AB}{=}\frac{1}{\sqrt{2}}(\ket{H}^A\ket{\phi_H}^B{+}\ket{V}^A\ket{\phi_V}^B). \label{Eq:nearest_pure_state}
\end{align}
In an ideal scenario of maximally entangled state, $\ket{\phi_H}$ and $\ket{\phi_V}$ are orthogonal to each other, i.e., $|\langle{\phi_H}|{\phi_V}\rangle|^2{=}0$. However, depending on the concurrence of our estimated nearest pure state, $\ket{\phi_V}$ will have a small contribution from $\ket{\phi_H}$. In our experiment, the concurrence of the estimated nearest pure state is $0.99$. This ensures that $\ket{\phi_H}$ and $\ket{\phi_V}$ are almost orthogonal, i.e. we have $|\langle{\phi_H}|{\phi_V}\rangle|^2{\approx}0$ . From Eq.~\eqref{Eq:nearest_pure_state}, we can see when Alice measures in $\{\ket{H},\ket{V}\}$, Bob gets maximum (anti-)correlation while measuring in $\{\ket{\phi_H},\ket{\phi_H^\perp}\}$ basis. 

Similarly, when Alice measures in a different basis, we can calculate the corresponding rotated mutually unbiased basis. For instance, when we express the nearest pure state in diagonal/anti-diagonal basis:
\begin{align}
    \ket{\psi}_\rho^{AB}{=}\frac{1}{\sqrt{2}}(\ket{D}^A\ket{\phi_D}^B{+}\ket{A}^A\ket{\phi_A}^B), \label{Eq:nearest_pure_diagonal}
\end{align}
we can see that when Alice is measuring in $\{\ket{D},\ket{A}\}$ basis, Bob has to measure in $\{\ket{\phi_D},\ket{\phi_D^\perp}\}$ basis to get the highest (anti-)correlation. Note that,
 $\ket{\phi_D}{=}1/\sqrt{2}(\ket{\phi_H}{+}\ket{\phi_V})$ and $\ket{\phi_A}{=}1/\sqrt{2}(\ket{\phi_H}{-}\ket{\phi_V})$. As discussed, the concurrence of our estimated nearest pure state being $0.99$ ensures $|\langle{\phi_A}|{\phi_D}\rangle|^2{\approx}0$. In the next section, we are going to discuss our experimental results.

\subsection{Experimental outcome} \label{sub_sec:exp_outcome}

We implement the BBM92 protocol using a polarization entangled bi-photon source (PEBS) based upon a doubly-pumped type-II SPDC process in a Sagnac loop\,\cite{fedrizzi_wavelength-tunable_2007}. An abridged schematic of our experimental setup has been provided in Fig.~\ref{fig:BBM92schematic_source}. Our source produces polarization entangled single photon pairs having $94\%$ fidelity with the Bell-state $\ket{\psi}_1{=}1/\sqrt{2}(\ket{HV}+\ket{VH})$, and a concurrence of $0.92$. We transmit the single photons to Alice and Bob module through two optical fibres. Each optical fibre is accompanied with a fibre-bench Polarization Controller Kit by Thorlabs (PC-FFB-780), see Fig.~\ref{fig:BBM92schematic_8dets}. For a detailed description of the experimental schematic, see Sec.~\ref{sub_sec:exp_schematic}. Our objective is to have a controlled polarization change introduced in the experiment to demonstrate the efficacy of our correction mechanism for different cases. The introduction of the polarization controller enables introduction of  polarization fluctuation in both single photons, to essentially manipulate/control the fidelity of the two-qubit state with the ideal $\ket{\psi_1}$ state at the output end. For each altered polarization state, we perform a quantum state tomography at the output. From the tomographically reconstructed density matrix, we estimate the nearest pure state using Lemma~\ref{Lemma:nearest_pure}. Next we evaluate the measurement basis at Bob's end giving maximum (anti-) correlation with Alice's $\sigma_3$ basis:$\{\ket{\phi_H},\ket{\phi_H^\perp}\}$ as in Eq.\eqref{Eq:nearest_pure_state} and with Alice's $\sigma_1$ basis: $\{\ket{\phi_D},\ket{\phi_D^\perp}\}$ as in Eq.\eqref{Eq:nearest_pure_diagonal}. To complete the protocol, one of the parties (say Alice) sends his/her time-stamp information to the other party (say Bob) via a publicly accessible classical channel. Bob then performs a cross correlation between his time-stamp and Alice's time-stamp to generate the coincidence peaks. He further runs the optimization algorithm \bibnote[optimization]{The details of the optimization algorithm will be presented elsewhere.} to optimize the window sizes for each coincidence peak to optimize the key rate and QBER. Based on the optimized window choices, Bob informs Alice, via the public classical channel, which of her time-stamps needs to be discarded. From the remaining time-stamps, the two parties can reconstruct their respective keys from the information about their measurement outcomes. Note that, these measurement outcomes are private to the individual parties.



In our result, we first show how optimizing measurement bases could result in low QBER irrespective of the polarization rotation through the single mode fibres. This is in stark contrast to the conventional approach, where we restrict our choice of measurement to Pauli bases. In such cases we would see higher QBER for lower fidelity. 

Next we use our optimization algorithm for further improvement of the performance of BBM92 protocol \bibnotemark[optimization]. To summarize, the goal of the optimization algorithm is twofold: reducing the overall QBER while maintaining a high key rate, and maintaining the QBER of the individual measurement bases below the information-theoretically secure bound, $11\%$.

We show advantages of finding the optimal measurement bases in Fig.~\ref{Fig:main_figure} where the orange circles represent QBERs for different fidelities with $\ket{\psi}_1$ for optimal measurement bases. We see how the optimal measurements can lower the QBERs below $11\%$ independent of the fidelity of the output state. In contrast, the blue circles represent the QBERs for different fidelities for conventional measurement bases ($\sigma_3$ and $\sigma_1$ bases). We note that the QBERs increase with lower fidelity. Moreover, the decrease of the QBERs is not monotonic, as in case of the fidelities in the range of $40-60\%$ and $70-80\%$. This is because, we are using the fixed Pauli bases of $\sigma_1$ and $\sigma_3$ for all the fidelity points. However, choosing a different Pauli bases turn out to be more optimal in certain scenarios, e.g., for a fidelity of $60\%$, choice of  $\sigma_2$ (instead of $\sigma_1$) and $\sigma_3$ bases are more optimal. The cyan crosses in Fig.~\ref{Fig:main_figure} convey the idea, where we use the tomographically reconstructed density matrix to estimate the QBERs for two optimal Pauli bases, and achieve a monotonically decreasing set of QBERs. In all such cases, the optimal choice of measurement bases outperforms conventional measurement bases. 

\begin{figure}[htbp] 
\begin{center}
\includegraphics[width=\columnwidth]{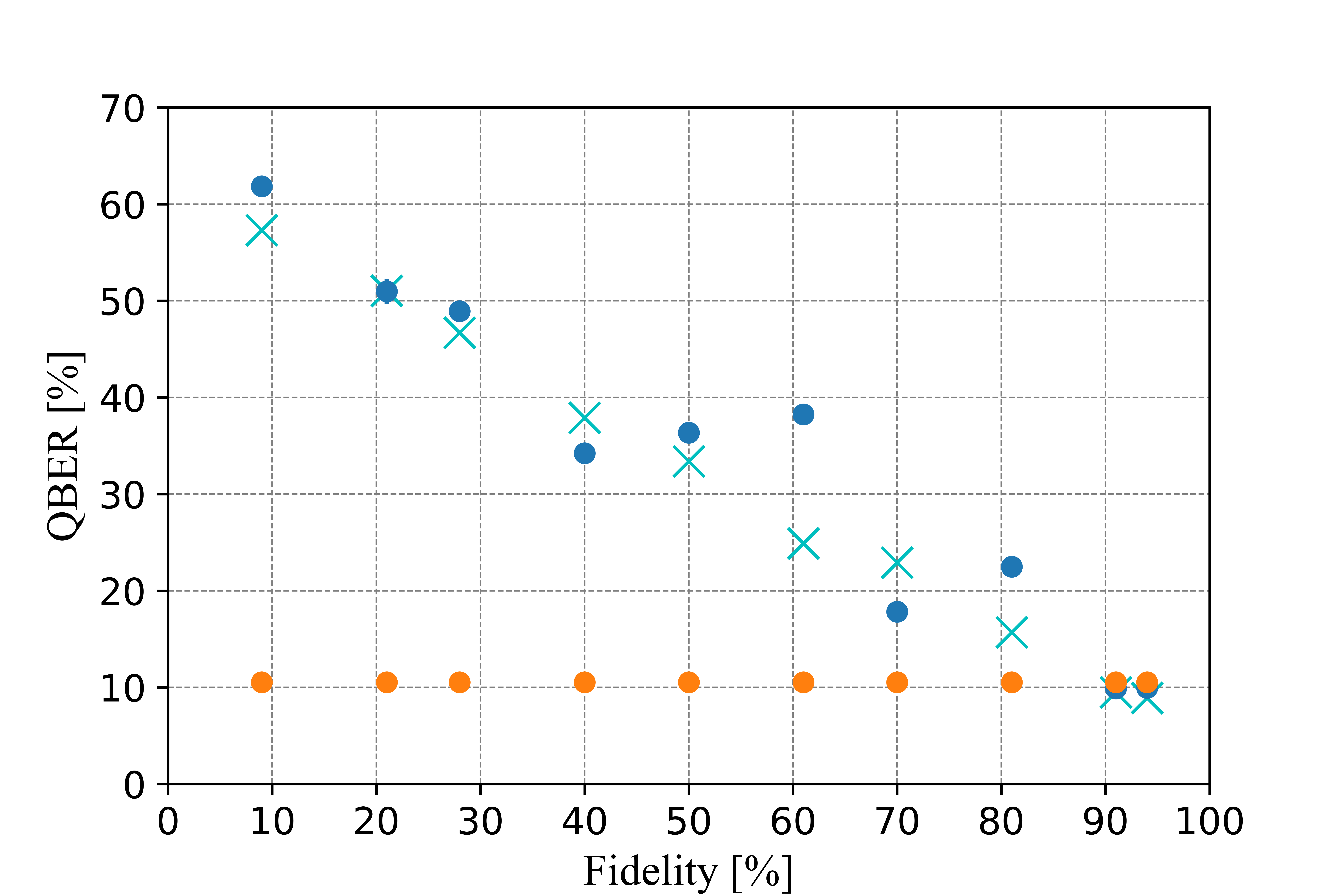}
\caption{Unoptimized QBERs (in $\%$) versus fidelity (in $\%$) with the singlet state $\ket{\psi}_1$. The blue circles represent our experimentally measured QBERs with the conventional bases of measurement (i.e., $\sigma_1$ and $\sigma_3$). It can be observed that the experimentally obtained QBERs are not monotonically decreasing, i.e., the QBERs portray an increasing trend for the fidelity of $40-60\%$ and $70-80\%$. However, a monotonic graph (i.e., the cyan crosses) can be obtained, when we choose the pair of Pauli measurement basis (i.e., any two out of $\sigma_1$, $\sigma_2$, and $\sigma_3$) that offers the best signal-to-noise ratio at each fidelity point. Nevertheless, it can be noted that the optimized measurement bases outperform the conventional measurement bases by offering a lower QBER (orange circles) irrespective of their fidelity with the singlet state. At each fidelity point, we measured 10 datasets, the mean of them represent the data points, while their standard deviation has been indicated with error bars.}
\vspace{-3mm}
\label{Fig:main_figure}
\end{center}
\end{figure}

In Fig.~\ref{Fig:main_figure_window}, we present the unoptimized results and show how varying coincidence window sizes lead to different QBERs and key rate. In both Figs.~\ref{Fig:qber_fid_unopt} and \ref{Fig:keyrate_fid_unopt}, the blue circles represent the data for 1 ns wide coincidence windows, the orange circles represent the data for 4 ns wide coincidence windows. We can see for 1 ns coincidence window, both QBERs ($\approx 5\%$) and key rate ($\approx 35$ Kbps) are lower compared to 4 ns, coincidence window where we have QBERs of $\approx 10\%$ and key rate of $\approx 50$ Kbps. To achieve optimal window sizes resulting in a better trade-off between QBER and key rate, i.e., low QBER and high key rate, we use the optimization algorithm. 

\begin{figure}[!ht]
	\centering
	\begin{subfigure}[b]{0.49\textwidth}
		\centering
		\includegraphics[width=\columnwidth]{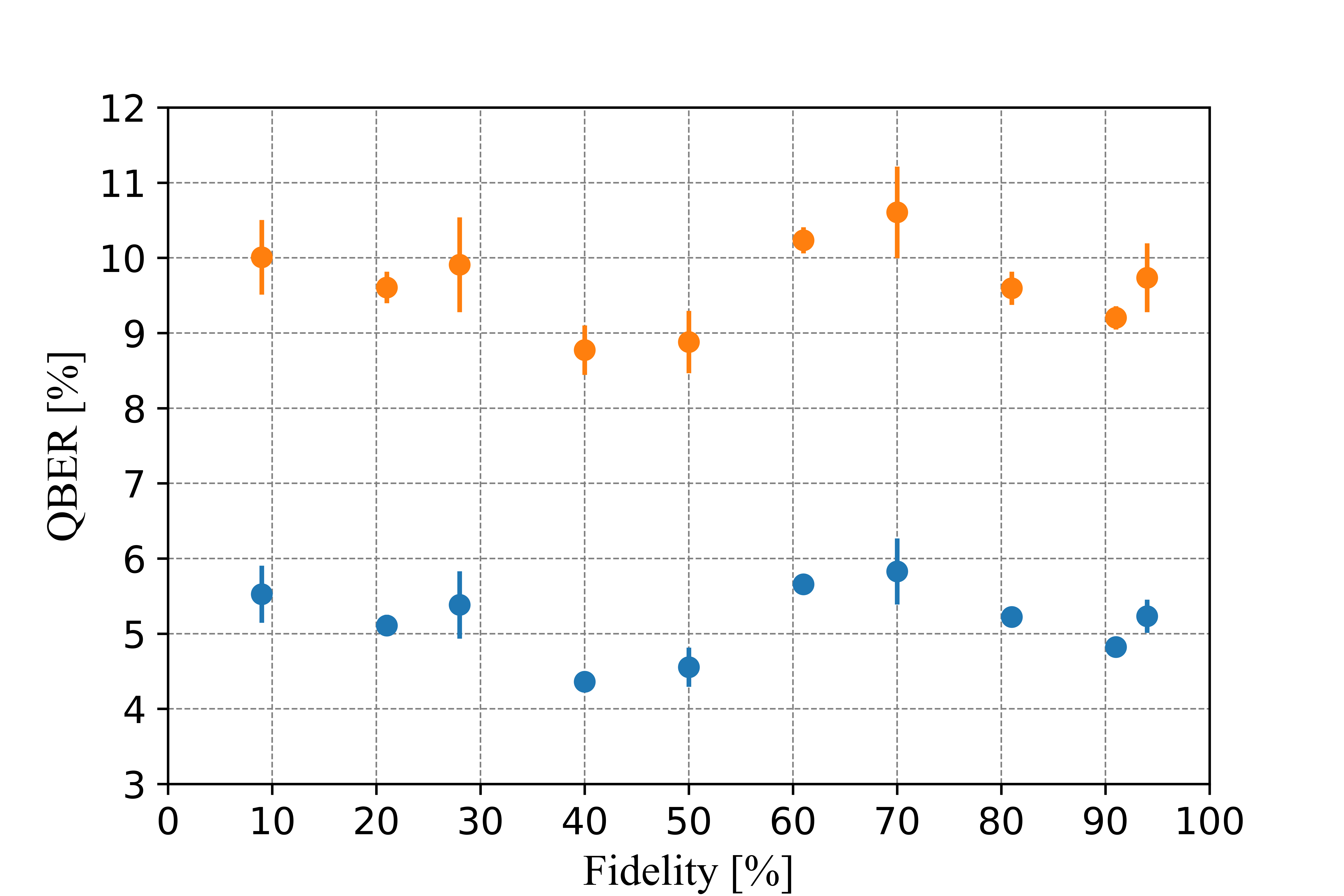}
		\caption{Unoptimized QBER (in \%) versus fidelity (in \%) with the $\ket{\psi}_1$ state plot for two different coincidence window spans.}
		\label{Fig:qber_fid_unopt} 
	\end{subfigure}
	\hfill
	\begin{subfigure}[b]{0.49\textwidth}
		\centering
		\includegraphics[width=\columnwidth]{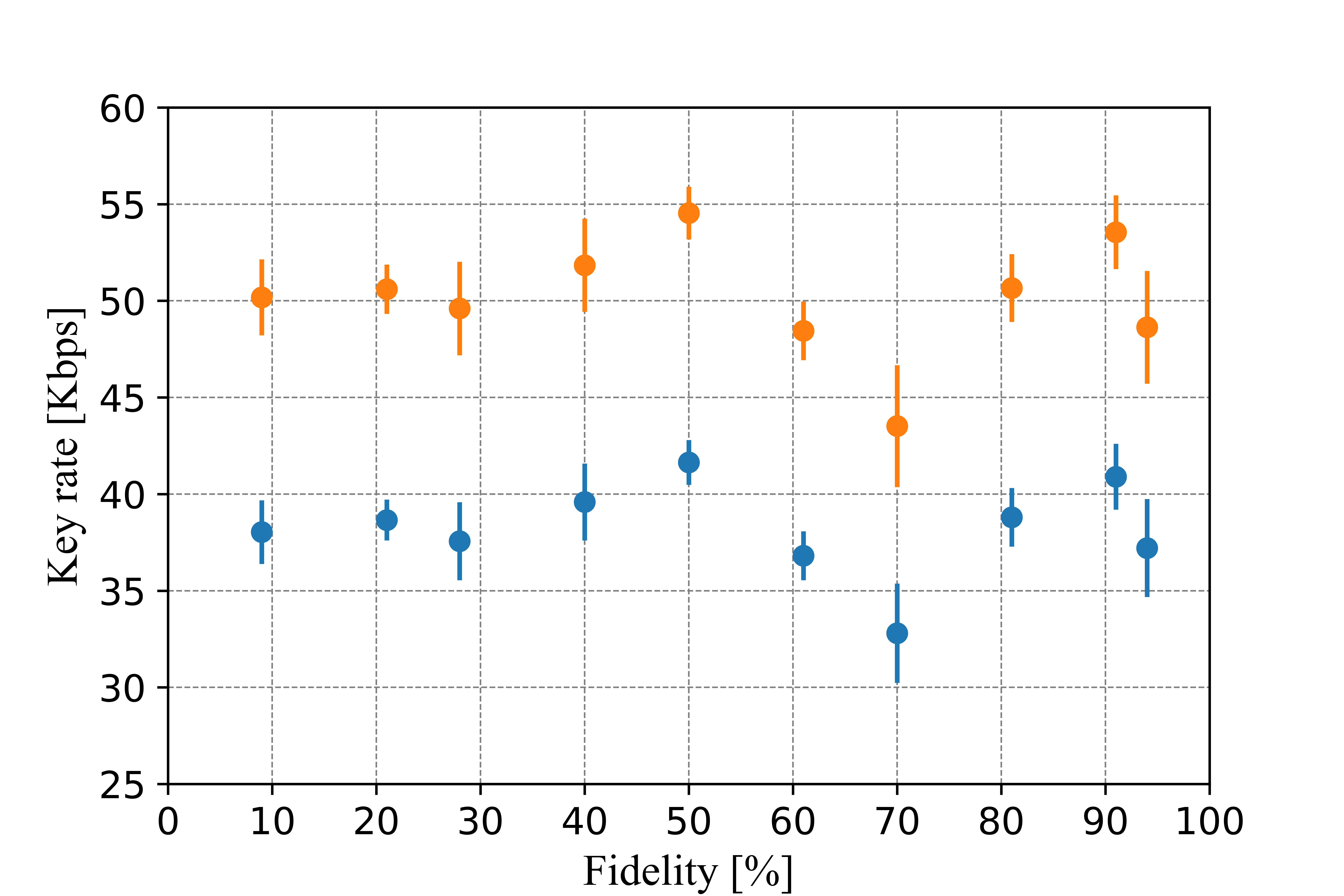}
		\caption{Unoptimized key rate (in Kbps) versus fidelity (in \%) with the $\ket{\psi}_1$ state plot for two different coincidence window spans.}
		\label{Fig:keyrate_fid_unopt}
	\end{subfigure}
	\caption{The blue and orange dots represent values corresponding to $1$ ns and $4$ ns coincidence window spans, respectively. For each fidelity point, the mean and standard deviation has been obtained over 10 measurement runs. The data points represent the mean of those runs, while the standard deviation in them has been indicated by the corresponding error bars.}
	\label{Fig:main_figure_window}
\end{figure}

In Fig.~\ref{Fig:overall_qber_fid_opt}, we show the advantages of the optimization algorithms. The orange circles represent optimized overall QBERs for the overall key string, however in such cases the QBERs for individual bases are not optimized. The blue circles represent the optimized overall QBERs where the individual QBERs are optimized as well. Using optimization, we could reduce the overall QBER while maintaining a high key rate of 40 Kbps. In Fig.~\ref{Fig:individual_qber_fid}, we show how the optimization algorithm takes QBERs for individual measurement bases into account. It is possible that while maintaining the overall QBER below $11\%$, the QBERs for individual measurement bases may shoot up above $11\%$ leading to leakage of information to the eavesdropper. To avoid this, it is important to contain the individual QBERs below $11\%$. In Fig.~\ref{Fig:individual_qber_fid}, the orange circles represent the maximum unoptimized QBERs for the individual measurement bases. The blue circles represent the optimized QBERs for the individual measurement bases. We note that the optimization algorithm ensures that the individual QBERs lie below $11\%$. 

\begin{figure}[!ht]
	\centering
	\begin{subfigure}[b]{0.49\textwidth}
		\centering
		\includegraphics[width=\columnwidth]{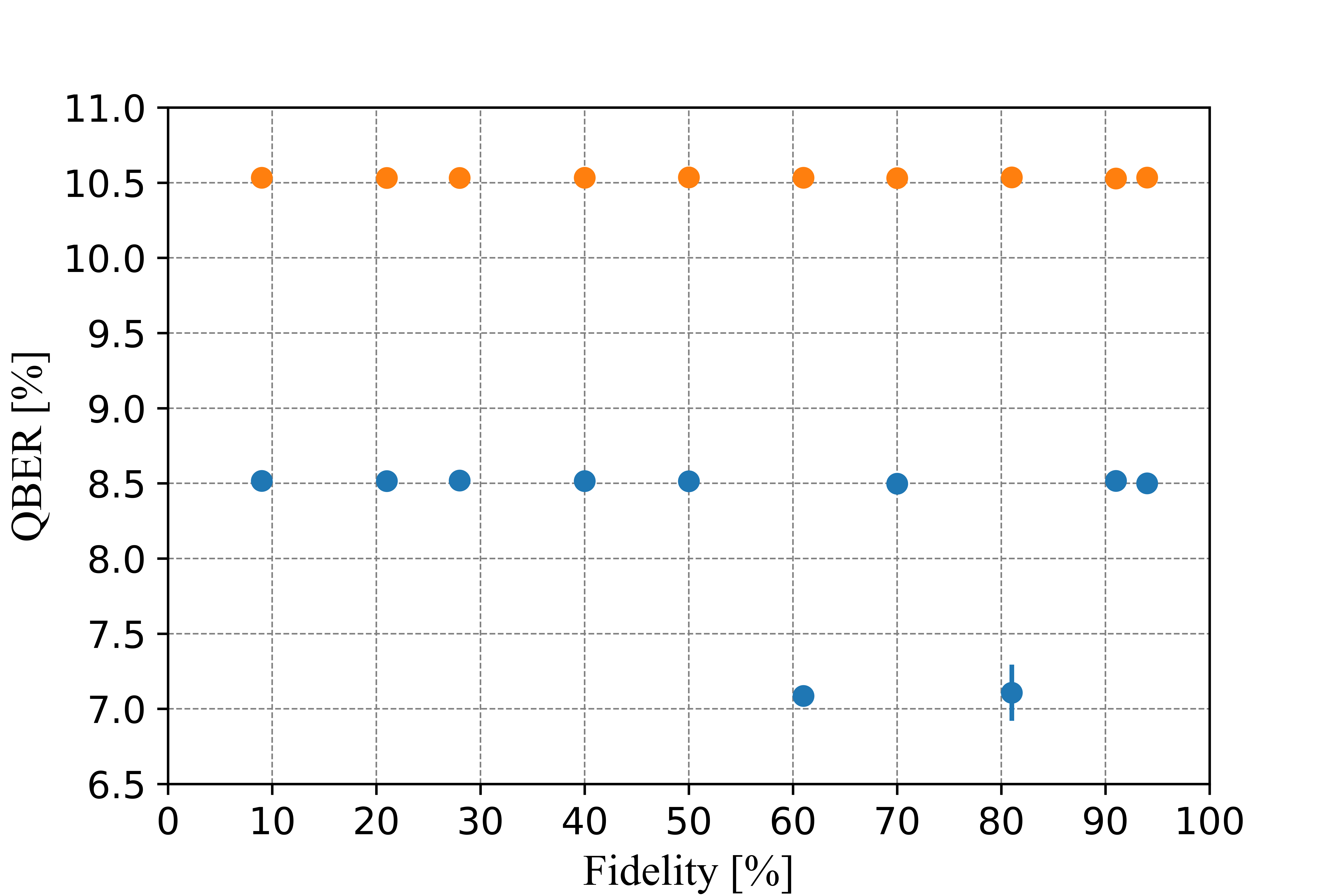}
		\caption{Optimized overall QBERs (in \%) considering coincidence measurements among both MUB versus fidelity (in \%) with the $\ket{\psi}_1$ state plot. The orange circles result from the optimization algorithm in which the QBERs for each individual MUB were not optimized, while the blue circles are obtained from another variant of the optimization algorithm in which they were optimized (restricted) to be below 11\%.}
		\label{Fig:overall_qber_fid_opt} 
	\end{subfigure}
	\hfill
	\begin{subfigure}[b]{0.49\textwidth}
		\centering
		\includegraphics[width=\columnwidth]{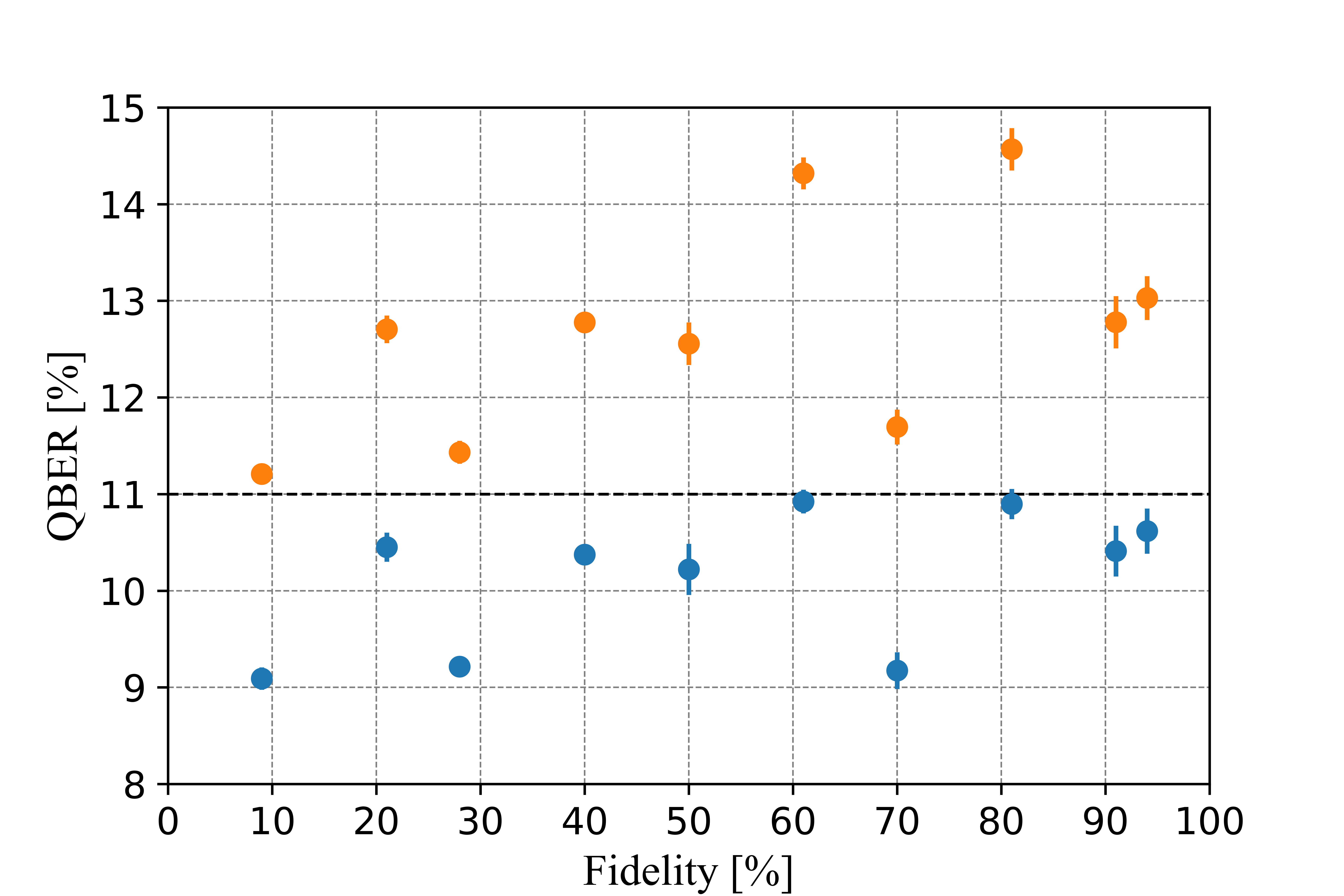}
		\caption{Maximal QBERs (in \%) considering coincidence measurements for a particular (i.e., individual) MUB versus fidelity (in \%) with the $\ket{\psi}_1$ state plot, corresponding to the two results depicted in Fig.~\ref{Fig:overall_qber_fid_opt}. Here again, the orange circles result from the optimization algorithm in which the QBERs for each individual MUB were not optimized, while the blue circles are obtained from another variant of the optimization algorithm in which they were optimized (restricted) to be below 11\% (as indicated with the black dashed line).}
		\label{Fig:individual_qber_fid}
	\end{subfigure}
	\caption{Plots highlighting the advantages of our optimization algorithms. It can be noted that by using our optimization algorithm, we could reduce the overall QBER while maintaining a high key rate of 40 Kbps. For each data points against a given fidelity, the mean and standard deviation has been obtained over 10 measurement runs similar to Fig.~\ref{Fig:main_figure_window} (although in these cases for optimized window spans). Again, the data points represent the mean of those runs, while the standard deviation in them has been indicated by the corresponding error bars.}
	\label{Fig:main_figure_2}
\end{figure}

\section{Discussion}
To conclude, we have addressed important practical challenges of long-distance QKD utilizing polarization of photons to encode the quantum state. The polarization state of the light is inevitably affected during long distance transmission. The conventional active feedback system based polarization tracking techniques are resource intensive, resulting in additional maintenance cost. We have shown that instead of active polarization correction, we can construct optimal measurement bases to achieve low QBER and high key rate irrespective of the polarization fluctuation. As a proof-of-principle demonstration, we have used the BBM92 protocol using polarization entangled photon pairs.  We overcome the polarization fluctuations of the single photons during transmission through optical fibres by performing optimal measurements. As our approach is based on the general principle of state preparation and measurement duality, the method can easily be extended to other QKD protocols. 
To construct the optimal choice of measurements, the parties perform a quantum state tomography on the received two qubit state before each QKD session. Based on the tomographically reconstructed density matrix, Bob arrives at his choice of measurement bases through the techniques introduced in Subsec.~\ref{subsec:theory}. Our approach overcomes the need for active feedback-based control systems. Another advantage of our protocol could be found in scenarios where the entangled photon source is itself not perfect. In principle, the generated entangled state could be either partially mixed or non-maximally entangled. Our technique provides a recipe to construct optimal measurement bases even in such non-ideal conditions. Our technique could be particularly advantageous in downlink-based QKD protcols, e.g., the quantum experiments using the Micius satellite~\cite{Micius_rev_mod} where the photon sources, being in the satellite, are not readily accessible to the experimentalists. In such cases, altering the more easily accessible measurement bases of only one party could overcome the detrimental effects of polarization fluctuation. So far we are constructing the optimal measurement bases using a tomographically complete dataset. As a future direction, it would be interesting to see if such optimization of measurement bases is possible without performing a full quantum state tomography. 

\section{Methods}
\subsection{Notations} \label{sec:notation} 
We introduce the relevant notations used in this paper. We denote the Pauli group as $\{\sigma_i|i{\in}\{0,1,2,3\}\}$ where $\sigma_0{=}\mathbb{1}$ is the identity operation and $\sigma_1$, $\sigma_2$ and $\sigma_3$ are Pauli $X$, $Y$ and $Z$ respectively. To denote the Bell states, note that all Bell states are related by local Pauli operations, hence we use an indexed notation of the Bell states as $\{\ket{\psi}_i|i{\in}\{0,1,2,3\}\}$, such that $\ket{\psi}_i{=}(\mathbb{1}{\otimes}\sigma_i)\ket{\psi}_0$. In this way, our indexed Bell states are 
\begin{align}
\ket{\psi}_0{=}\frac{1}{\sqrt{2}}(\ket{00}{+}\ket{11}), \  \
\ket{\psi}_1{=}\frac{1}{\sqrt{2}}(\ket{01}{+}\ket{10}),\nonumber\\
\ket{\psi}_2{=}\frac{1}{\sqrt{2}}(\ket{01}{-}\ket{10}), \ \ 
\ket{\psi}_3{=}\frac{1}{\sqrt{2}}(\ket{00}{-}\ket{11}).
\end{align}
As we use polarization degree of freedom of the photons to encode the quantum state --- horizontal polarization $\ket{H}{\rightarrow}\ket{0}$, and vertical polarization $\ket{V}\rightarrow\ket{1}$, --- we will interchangeably use standard notations for polarization  to denote our quantum states:
\begin{align}
\ket{H}&{\equiv}{\ket{0}},\,\, \ket{V}{\equiv}{\ket{1}}, \nonumber \\
\ket{D}&{\equiv}\frac{1}{\sqrt{2}}(\ket{0}{+}\ket{1}),\,\, \ket{A}{\equiv}\frac{1}{\sqrt{2}}(\ket{0}{-}\ket{1}), \nonumber \\
\ket{R}&{\equiv}\frac{1}{\sqrt{2}}(\ket{0}{+}i\ket{1}),\,\,
\ket{L}{\equiv}\frac{1}{\sqrt{2}}(\ket{0}{-}i\ket{1}). 
\end{align}

\subsection{Experimental schematic}\label{sub_sec:exp_schematic}

Our experimental setup for implementing the BBM92 protocol contains a PEBS, which generates polarization entangled photons pairs via spontaneous parametric down-conversion (SPDC) process, from a doubly-pumped type-II periodically-poled KTP (PPKTP) crystal in a Sagnac configuration as schematically represented in Fig.~\ref{fig:BBM92schematic_source}\,\cite{fedrizzi_wavelength-tunable_2007}. More specifically, we use a 405 nm pump beam through a type-II PPKTP crystal with 10 $\mu$m poling period as indicated in Fig.~\ref{fig:BBM92schematic_source}. Passing the pump beam through the crystal produces down-converted, degenerate single photon pairs with central wavelength of 810 nm. The single photons are then separated from the pump beam via two dichroic mirrors. In case of the perfect alignment of the setup, a horizontally polarized ($\ket{H}$) pump beam produces two down-converted single photons of polarization state $\ket{HV}$. A vertically polarized ($\ket{V}$) pump beam, on the other hand, produces single photons with polarization state $\ket{VH}$. By changing the pump polarization, we can get a polarization entangled state $\ket{\psi}_{\phi}{=}1/\sqrt{2}(\ket{HV}{+}e^{i\phi}\ket{VH})$. The relative phase $\phi$ can be (adjusted) set to zero by varying the pump polarization using the quater-wave plate (QWP) and half-wave plate (HWP) placed at the entry of the interferometer, i.e., in principle oriented at $0^{\circ}$ and $22.5^{\circ}$, respectively, in order to produce at singlet state as shown in Fig.~\ref{fig:BBM92schematic_source}. To characterize our state, we perform a quantum state tomography at the output of the source. Our source has a $91\%$ purity, $94\%$ fidelity with respect to $\ket{\psi}_1{=}1/\sqrt{2}(\ket{HV}{+}\ket{VH})$, and a concurrence of $0.92$. After developing the source, we dispatch the photons through two optical fibres to two setups (each approximately 5 meter apart from the source), referred as Alice and Bob modules as introduced in \ref{sub_sec:exp_outcome}. We show the schematic of the experimental setup for the BBM92 measurement scheme in Fig.~\ref{fig:BBM92schematic_8dets}.

\begin{center}
\begin{figure}[!h]
\centering
\includegraphics[width=\columnwidth]{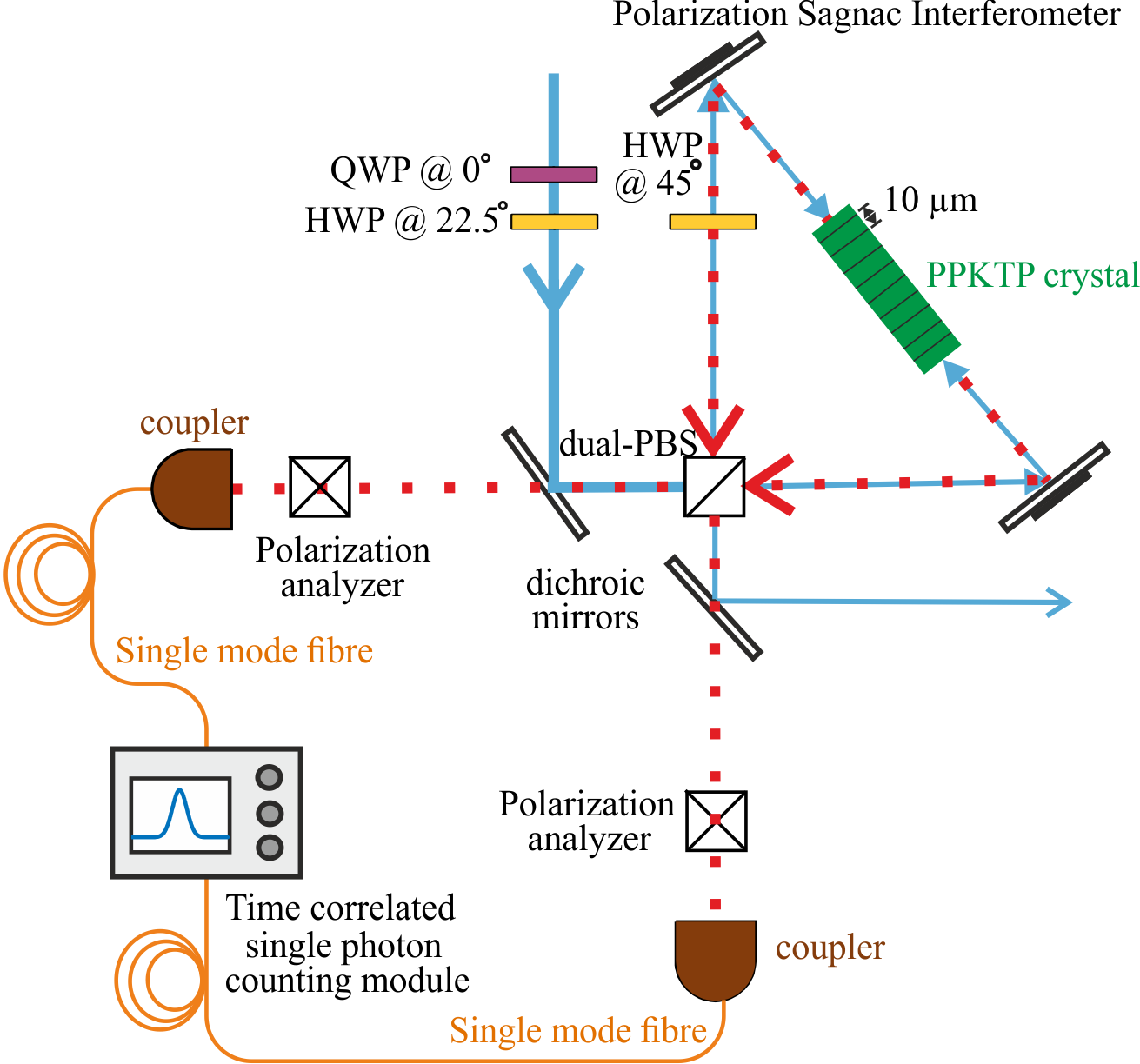}
\caption{Schematic of the Sagnac interferometer based polarization entangled bi-photon source (PEBS) setup.}
\label{fig:BBM92schematic_source}
\end{figure}
\par\end{center}

\begin{center}
\begin{figure}[!h]
\centering
\includegraphics[width=\columnwidth]{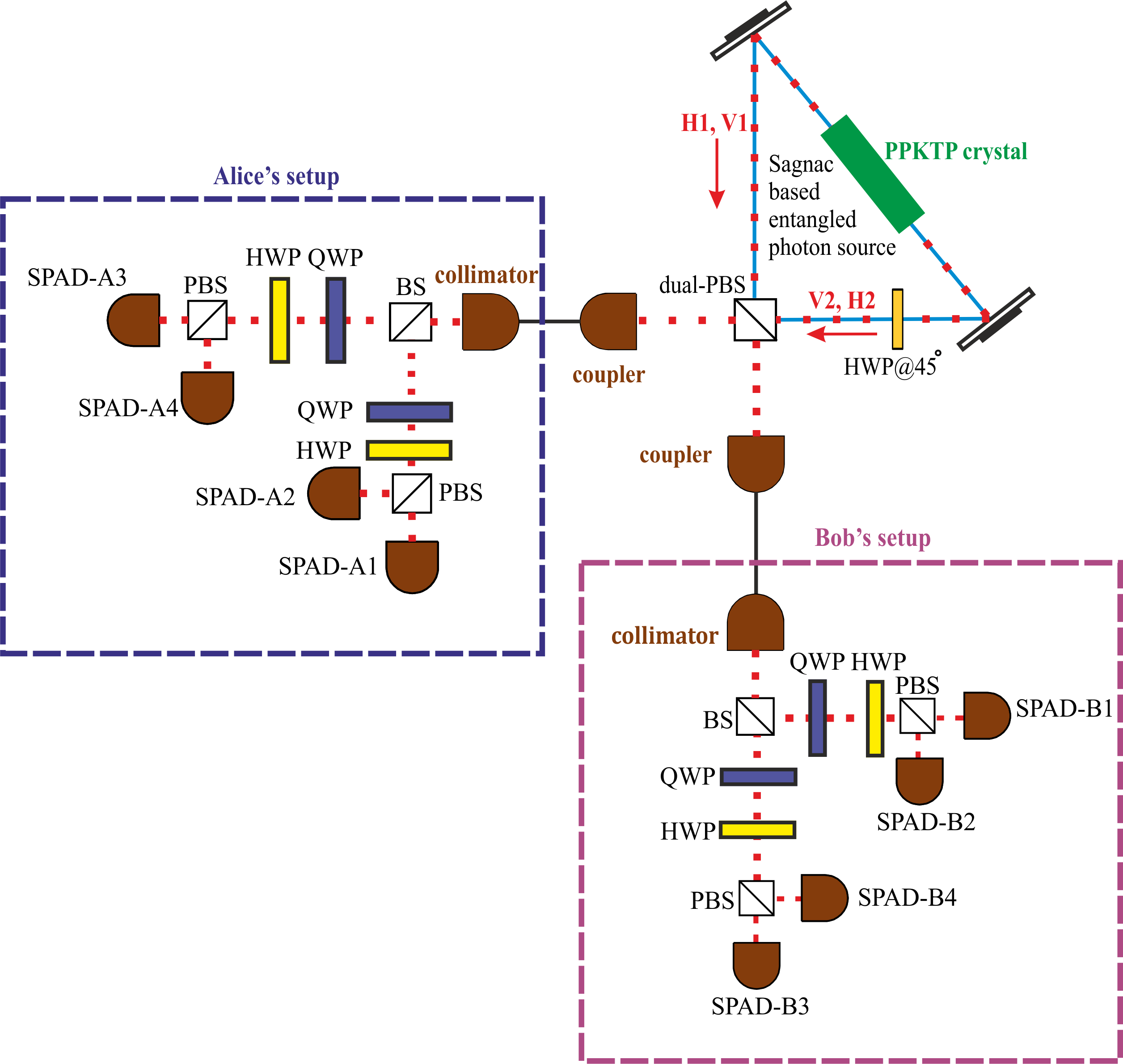}
\caption{Schematic of the experimental setup used for realizing the BBM92 protocol. In this setup, the polarization-entangled photon-pairs are transferred to the eight single-photon avalanche detectors (SPADs) to implement the relevant measurement bases for the protocol. The paired terms (H1, V1) and (V2, H2), indicated in red, represent the corresponding polarization of the daughter photons emerging, from each pump photon striking the crystal, in two different directions.}
\label{fig:BBM92schematic_8dets}
\end{figure}
\par\end{center}

During transmission, the polarization of the single photons at Alice and Bob modules are affected. To mitigate this, in Alice module, we randomly measure the stream of incoming source-photons along the rectilinear and diagonal projection bases. On the other hand, in Bob module, we randomly measure the polarization of incoming source-photons along the $\{\ket{\phi_H},\ket{\phi_H^{\perp}}\}$ and $\{\ket{\phi_D},\ket{\phi_D^\perp}\}$ projection bases. We implement the random choice of measurements using 50:50 beam-splitters.Each of these basis projections can lead to either of the two outcomes (detection of the photon along the transmitted arm (H/D), and ($\phi_H$/$\phi_D$) or other along the reflected arm (V/A), and ($\phi_H^{\perp}$/$\phi_D^{\perp}$) of the PBS). In this way, we end up with a total of eight coincidence detection between the Alice's and Bob's detector clicks. Four of those form the desirable set (signal) and the other four form the undesirable set (noise). We assess the number of coincidences in these sets by analysing the signal-to-noise ratios (SNRs), through the consideration of suitable window spans around the peak maxima. The signal and the noise values thus obtained within these window regions are then used to compute the raw key rate, and the quantum-bit-error-rate (QBER) of BBM92 protocol implementation.

\section*{ACKNOWLEDGEMENTS}

US would like to thank the Indian Space Research Organization for support through the QuEST-ISRO research grant. We thank Kaushik Joarder for initial assistance in photon source design.

\bibliographystyle{apsrev4-1}
\bibliography{manuscript.bib}

\begin{thebibliography}{43}
\bibitem [{opt()}]{optimization}%
  \BibitemOpen
  \href@noop {} {}\bibinfo {note} {The details of the optimization algorithm
  will be presented elsewhere.}\BibitemShut {Stop}%
\end{thebibliography}

\end{document}